\documentclass[journal]{IEEEtran}

\usepackage{amsmath, amsthm, amsfonts}
\usepackage{tikz}

\usetikzlibrary{matrix}
\usetikzlibrary{decorations.pathreplacing}

\newcommand{\ket}[1]{|{#1}\rangle}
\newcommand{\bra}[1]{\langle{#1}|}

\DeclareMathOperator{\Tr}{Tr}
\DeclareMathOperator{\wt}{wt}

\newtheorem{proposition}{Proposition}
\newtheorem{theorem}[proposition]{Theorem}

\title{A Combinatorial Interpretation for the Shor-Laflamme Weight Enumerators of CWS Codes}
\author{Andrew~Nemec
        ~and~Andreas~Klappenecker
\thanks{This research was supported in part by a Texas A\&M University T3 grant.}%
\thanks{A. Nemec and A. Klappenecker are with the Department
of Computer Science \& Engineering, Texas A\&M University, College Station,
TX, 77843 USA e-mail: (nemeca@tamu.edu; klappi@cse.tamu.edu).}}

\begin{document}

\maketitle

\begin{abstract}
We show that one of the Shor-Laflamme weight enumerators of a codeword stabilized quantum code may be interpreted as the distance enumerator of an associated classical code.
\end{abstract}

\begin{IEEEkeywords}
Quantum error-correcting codes, codeword stabilized codes, nonadditive codes, Shor-Laflamme weight enumerators.
\end{IEEEkeywords}

\IEEEpeerreviewmaketitle

\section{Introduction}

The weight enumerators of classical error-correcting codes arise in the derivation of the upper bounds for code parameters via linear programming \cite{Delsarte1972}. These weight enumerators admit a combinatorial interpretation, as for additive codes they count the number of codewords of each weight in the code, and more generally for nonadditive codes they describe the distances between each pair of codewords. The weight enumerators of a code and its dual code are connected by the MacWilliams identity \cite{Macwilliams1963}.

Shor and Laflamme defined a pair of weight enumerators for quantum codes \cite{Shor1997} which were used by Ashikhmin and Litsyn to develop linear programming bounds for the parameters of quantum codes \cite{Ashikhmin1999}. Rains also defined the similar unitary enumerators \cite{Rains1998}, as well as the quantum analogue of the shadow enumerators \cite{Rains1999c} which provide sharper linear programming bounds when used in conjunction with the Shor-Laflamme weight enumerators. Weight enumerators have also been used to derive linear programming bounds for other quantum code variants, such as subsystem codes \cite{Aly2006}, asymmetric quantum codes \cite{Sarvepalli2009}, hybrid codes \cite{Grassl2017}, entanglement-assisted codes \cite{Lai2018}, and quantum amplitude damping codes \cite{Ouyang2020}.

While the Shor-Laflamme weight enumerators do not in general have a known combinatorial interpretation similar to the weight enumerators of classical codes, they do for the well-known class of stabilizer codes. The $A\!\left(z\right)$ Shor-Laflamme weight enumerator counts the number of elements of each weight in the stabilizer group associated with the stabilizer code, while the $B\!\left(z\right)$ Shor-Laflamme weight enumerator does the same for elements in the centralizer of the stabilizer group \cite{Gottesman1997}. These correspond to the weight enumerators of a classical self-orthogonal additive code and its dual code.

For non-stabilizer quantum codes (known as nonadditive quantum codes), little is known about the weight enumerators. The unitary enumerators may be interpreted as the binomial moments of the distance distribution of classical codes \cite{Ashikhmin2000}, and the Shor-Laflamme weight enumerator $A\!\left(z\right)$ can be interpreted as the two-norms of the $j$-body correlations of the code \cite{Huber2017}, but their remains no combinatorial interpretation of the Shor-Laflamme weight enumerator $B\!\left(z\right)$. In this paper, we show that for the nonadditive codeword stabilized codes \cite{Cross2009}, the Shor-Laflamme weight enumerator $B\!\left(z\right)$ may be interpreted as the distance enumerator of an associated nonadditive classical code, partially answering a question recently posed by Ball, Centelles, and Huber \cite{Ball2020}.

\section{Background}

For classical codes, the distance between codewords is given by the Hamming distance: $$d_{H}\!\left(x,y\right)=\left\lvert\left\{i\mid x_{i}\neq y_{i}\right\}\right\rvert.$$ The distance distribution $A$ of an $\left(n,M,d\right)$ classical code $C$ is a vector of length $\left(n+1\right)$, where $$A_{i}=\frac{1}{M}\left\lvert\left\{\left(x,y\right)\mid x,y\in C,d_{H}\!\left(x,y\right)=i\right\}\right\rvert,$$ meaning that $A_{i}$ is the number of codewords at distance $i$ from each other, normalized by the size of the code. The polynomial $$A\!\left(z\right)=\sum\limits_{i=0}^{n}A_{i}z^{i}$$ is the distance enumerator of the code. The minimum distance $d$ of the code is the smallest index $i\neq0$ such that $A_{i}$ is non-zero.

The Hamming weight of a codeword is the distance from the all zero codeword, that is $\wt_{H}\!\left(x\right)=d_{H}\!\left(x,0^{n}\right)$. If $C$ is an additive code, that is a code which is closed under addition, then $A$ counts the number of codewords of each weight, so $$A_{i}=\left\lvert\left\{x\mid x\in C,\wt\!\left(x\right)=i\right\}\right\rvert,$$ and we call $A$ the weight distribution and $A\!\left(z\right)$ the weight enumerator of the code. The weight enumerator of an additive code $C$ is connected to the weight enumerator $B\!\left(z\right)$ of its dual code $C^{\perp}$ by the MacWilliams identity \cite{Delsarte1972, Macwilliams1963}: $$B\!\left(z\right)=\frac{\left(1+z\right)^{n}}{M}A\!\left(\frac{1-z}{1+z}\right).$$ For a nonadditive code, the MacWilliams identity may still be formally defined in the same way, although the resultant polynomial in general does not correspond to the distance enumerator of any code \cite{Macwilliams1972, Cohen1997}.

An $\left(\!\left(n,K,d\right)\!\right)$ quantum code $\mathcal{C}$ on $n$ physical qubits is a $K$-dimensional subspace of the Hilbert space $\mathbb{C}^{2^{n}}$. Let $X$ and $Z$ be the Pauli operators $$X=\begin{pmatrix} 0 & 1 \\ 1 & 0 \end{pmatrix}\text{ and }Z=\begin{pmatrix} 1 & 0 \\ 0 & -1 \end{pmatrix}.$$ A basis for the linear operators of the Hilbert space can be given by tensor products of the Pauli operators: $$\mathcal{E}_{n}=\left\{E_{1}\otimes\cdots\otimes E_{n}\mid E_{i}=X^{a_{i}}Z^{b_{i}}, a_{i},b_{i}\in\mathbb{F}_{2}\right\}.$$ Each element $E\in\mathcal{E}_{n}$ can be associated with a unique codeword $\left(a\mid b\right)=\left(a_{1},\dots,a_{n}\mid b_{1},\dots,b_{n}\right)$ of length $2n$. The distance between two codewords of this type is given by the symplectic distance: $$d_{s}\!\left(\left(a\mid b\right),\left(a'\mid b'\right)\right)=\left\lvert\left\{k\mid\left(a_{k},b_{k}\right)\neq\left(a'_{k},b'_{k}\right)\right\}\right\rvert.$$ The symplectic weight $\wt\!\left(E\right)$ is the number of non-identity tensor components $E_{i}$ make up $E$.

The most well studied class of quantum codes are the stabilizer codes \cite{Gottesman1997}. An $\left[\!\left[n,k,d\right]\!\right]$ stabilizer code is defined by its stabilizer group $\mathcal{S}$, which is generated by $n-k$ mutually commuting independent operators $S_{i}\in G_{n}$ (which does not include $-I$), where $$G_{n}=\left\{i^{\ell}E\mid \ell\in\mathbb{Z}_{4},E\in\mathcal{E}_{n}\right\}$$ is the error group on $n$ qubits. The stabilizer code is then the $2^{k}$-dimensional joint $+1$-eigenspace of $\mathcal{S}$. Associated with the stabilizer group is its centralizer in $C_{G_{n}}\!\left(\mathcal{S}\right)$, the group of all elements in $G_{n}$ that commute with every element in $\mathcal{S}$. These are the operators that act as the logical operators on the encoded states of the code.

Shor and Laflamme \cite{Shor1997} defined a pair of weight enumerators $A\!\left(z\right)$ and $B\!\left(x\right)$ for quantum codes in the following fashion: \begin{equation*} A_{i}=\frac{1}{K^{2}}\sum\limits_{\substack{E\in\mathcal{E}_{n} \\ \wt\left(E\right)=i}}\Tr\!\left(EP\right)\Tr\!\left(E^{*}P\right) \end{equation*} and \begin{equation*} B_{i}=\frac{1}{K}\sum\limits_{\substack{E\in\mathcal{E}_{n} \\ \wt\left(E\right)=i}}\Tr\!\left(EPE^{*}P\right), \end{equation*} where $P$ is the orthogonal projector onto the code $\mathcal{C}$. In general, the weight enumerators of quantum codes do not seem to admit as nice a combinatorial interpretation as they do for classical codes. However, for stabilizer codes there is such an interpretation, as $A\!\left(z\right)$ counts the number of elements of each weight in the stabilizer group $\mathcal{S}$ and $B\!\left(z\right)$ counts the number of elements in the centralizer $C_{G_{n}}\!\left(\mathcal{S}\right)$ (modulo the phases on the Pauli elements). Additionally, each element of the stabilizer and centralizer can be associated with a unique (up to phase) codeword of length $2n$. Let $C$ be the code containing the set of codewords associated with $\mathcal{S}$. Then its symplectic dual $C^{\perp}$ is the code associated with $C_{G_{n}}\!\left(\mathcal{S}\right)$. Additionally, since $\mathcal{S}\leq C_{G_{n}}\!\left(\mathcal{S}\right)$ (as $\mathcal{S}$ is Abelian), we have that $C\subseteq C^{\perp}$, that is $C$ is self-orthogonal.

\section{Weight Enumerators of CWS Codes}

Codeword stabilized (CWS) codes were introduced by Cross et al. \cite{Cross2009} as a framework to construct quantum codes that includes all stabilizer codes and most nonadditive codes with good parameters, for example, see \cite{Rains1997, Yu2007, Yu2008}. A CWS code is comprised of two objects: a stabilizer group $\mathcal{S}$ generated by $n$ mutually commuting independent elements of $G_{n}$ (not including $-I$), so $\mathcal{S}$ stabilizes a single stabilizer state $\ket{\varphi}$ which comprises a $1$-dimensional stabilizer code, and a collection $\mathcal{T}$ of $K$ commuting codeword operators $T_{i}\in G_{n}$ such that each $T_{i}$ is from a separate coset of $G_{n}/\mathcal{S}$. Without loss of generality, we can choose $T_{1}=I$.

The set $\mathcal{TS}=\left\{i^{\ell}T_{i}S_{j}\mid T_{i}\in\mathcal{T},S_{j}\in\mathcal{S},\ell\in\mathbb{Z}_{4}\right\}$ plays a similar role to the centralizer of a stabilizer code, and a CWS code is a stabilizer code precisely when $\mathcal{TS}$ forms an additive group. We show now that the Shor-Laflamme weight enumerator $B\!\left(z\right)$ of a CWS code is the distance enumerator of the classical code associated with the set $\mathcal{TS}$.

\begin{theorem}
Let $\mathcal{C}$ be a CWS quantum code. Let $C$ be the classical symplectic code associated with the set $\mathcal{TS}$. Then the Shor-Laflamme weight enumerator $B\!\left(z\right)$ of the quantum code $\mathcal{C}$ is the distance enumerator of the classical code $C$.
\end{theorem}
\begin{proof}
Let $P$ be the projector onto $\mathcal{C}$, $\mathcal{S}$ be the stabilizer group of the stabilizer state $\ket{\varphi}$, and $\mathcal{T}=\left\{T_{1},T_{2},\dots,T_{K}\right\}$ be the set of codeword operators. We can write $P$ as \begin{equation*}
    P=\sum\limits_{i=1}^{K} T_{i}\ket{\varphi}\bra{\varphi}T_{i}^{*}.
\end{equation*}

We can expand the weight enumerator as \begin{align*}
    B_{d} & = \frac{1}{K}\sum\limits_{\substack{E\in\mathcal{E}_{n} \\ \wt\left(E\right)=d}}\Tr\!\left(EPE^{*}P\right) \\
    & = \frac{1}{K}\sum\limits_{\substack{E\in\mathcal{E}_{n} \\ \wt\left(E\right)=d}}\sum\limits_{i,j=1}^{K}\left\lvert\bra{\varphi}T_{j}^{*}E T_{i}\ket{\varphi}\right\rvert^{2}.
\end{align*}

Since elements of the Pauli group either commute or anticommute with each other, it follows that $\left\lvert\bra{\varphi}t_{j}^{*}E t_{i}\ket{\varphi}\right\rvert^{2}=1$ for all error operators $E\in T_{i}^{*}T_{j}\mathcal{S}$. Furthermore, if $E\notin T_{i}^{*}T_{j}\mathcal{S}$, we may write $E=T'F$, where $F\in\mathcal{S}$ and $T'\neq T_{i}^{*}T_{j}=T_{i}T_{j}^{*}$ is a coset representative of $T'\mathcal{S}\neq T_{i}^{*}T_{j}\mathcal{S}$. Then \begin{align*} \bra{\varphi}T_{j}^{*}E T_{i}\ket{\varphi} & = \pm\bra{\varphi}T_{j}^{*}T_{i}T'F\ket{\varphi} \\ & = \pm\bra{\varphi}T_{j}^{*}T_{i}T'\ket{\varphi} \\ & = 0, \end{align*} as there must be a stabilizer element $s\in\mathcal{S}$ that anticommutes with $T_{i}^{*}T_{j}T'$, and so \begin{align*} \pm\bra{\varphi}T_{j}^{*}T_{i}T'\ket{\varphi} & = \pm\bra{\varphi}sT_{j}^{*}T_{i}T'\ket{\varphi} \\ & = \mp\bra{\varphi}T_{j}^{*}T_{i}T's\ket{\varphi} \\ & = \mp\bra{\varphi}T_{j}^{*}T_{i}T'\ket{\varphi},\end{align*} implying that $\bra{\varphi}T_{j}^{*} T_{i}T'\ket{\varphi}=0$. Therefore, we have that \begin{equation*}
    \left\lvert\bra{\varphi}T_{j}^{*}E T_{i}\ket{\varphi}\right\rvert^{2} = \begin{cases}
    1, & E\in T_{i}^{*}T_{j}\mathcal{S} \\
    0, & \text{otherwise}
    \end{cases}
\end{equation*}

Let $S$ be the set of $2^{n}$ classical codewords associated with the elements of $\mathcal{S}$, $T=\left\{t_{1},t_{2},\dots,t_{k}\right\}$ the set of classical codewords associated with $\mathcal{T}$, and $F$ a set of $2^{n}$ classical codewords such that $T\subset F$ and $FS=\left\{0,1\right\}^{2n}$. Any codeword may be written as $f+s$, where $f\in F$ and $s\in S$. Given two pairs of codewords $c=\left(f+s,g+u\right),c'=\left(f'+s',g'+u'\right)\in\left\{0,1\right\}^{4n}$, we define the equivalence relation $\sim$ such that $c\sim c'$ if and only if $s+u=s'+u'$, $f=f'$, and $g=g'$. It is straightforward to check that $\sim$ is indeed a reflexive, symmetric, and transitive relation, and therefore an equivalence relation.

Since $S$ is a normal subgroup of $\left(FS\right)^{2}$, the quotient group $\left(FS\right)^{2}/S$ is isomorphic to $F^{2}S$. Denote elements of this set by $f+g+w$, where $f,g\in F$, $w\in S$. For all pairs of codewords $\left(f+s,g+u\right)$ in the same partition, $d_{s}\!\left(f+s,g+u\right)=\wt_{s}\!\left(f+g+w\right)$, where $w=s+u$. This means that for the classical code $C=TS$, the distance distribution \begin{align*}
    A_{d} = & \frac{1}{K2^{n}}\left\lvert\left\{\left(t_{i}+s,t_{j}+u\right)\mid d_{s}\!\left(t_{i}+s,t_{j}+u\right)=d\right\}\right\rvert \\
    = & \frac{1}{K}\lvert\{t_{i}+t_{j}+w\mid t_{i},t_{j}\in T, w\in S, \\ & \;\;\;\;\;\;\wt_{s}\!\left(t_{i}+t_{j}+w\right)=d\}\rvert.
\end{align*}

We associate $t_{i}$, $t_{j}$, and $w$ with the quantum operators $T_{i}$, $T_{j}$, and $W$. Note that $\left\lvert\bra{\varphi}T_{j}^{*}E T_{i}\ket{\varphi}\right\rvert^{2}=1$ if and only if $E=T_{i}^{*}T_{j}W$, meaning that the weight distribution of the quantum code is identical to the distance distribution of the classical code. In the case that there are two (or more) pairs of codeword operators such that $T_{i}T_{j}=T_{k}T_{\ell}$, the operator $E$ might be counted twice by both $\left\lvert\bra{\varphi}T_{j}^{*}E T_{i}\ket{\varphi}\right\rvert^{2}$ and $\left\lvert\bra{\varphi}T_{\ell}^{*}E T_{k}\ket{\varphi}\right\rvert^{2}$, but this is offset by $E$ not being checked separately as $E=T_{i}^{*}T_{j}W$ and $E=T_{k}^{*}T_{\ell}W$.

This shows that the Shor-Laflamme weight enumerator $B\!\left(z\right)$ of the quantum code $\mathcal{C}$ is the same as the distance enumerator of its associated classical code $C$.
\end{proof}

\begin{figure}[t]
\centering
\begin{equation*}\small
G=\left(
\begin{tikzpicture}[baseline=-.5ex]
\matrix[
  matrix of math nodes,
  column sep=.25ex, row sep=-.25ex
] (m)
{
1 & 0 & 0 & 0 & 0 & 0 & 0 & 0 & 0 & 0 & 1 & 0 & 0 & 0 & 0 & 0 & 0 & 1 \\
0 & 1 & 0 & 0 & 0 & 0 & 0 & 0 & 0 & 1 & 0 & 1 & 0 & 0 & 0 & 0 & 0 & 0 \\
0 & 0 & 1 & 0 & 0 & 0 & 0 & 0 & 0 & 0 & 1 & 0 & 1 & 0 & 0 & 0 & 0 & 0 \\
0 & 0 & 0 & 1 & 0 & 0 & 0 & 0 & 0 & 0 & 0 & 1 & 0 & 1 & 0 & 0 & 0 & 0 \\
0 & 0 & 0 & 0 & 1 & 0 & 0 & 0 & 0 & 0 & 0 & 0 & 1 & 0 & 1 & 0 & 0 & 0 \\
0 & 0 & 0 & 0 & 0 & 1 & 0 & 0 & 0 & 0 & 0 & 0 & 0 & 1 & 0 & 1 & 0 & 0 \\
0 & 0 & 0 & 0 & 0 & 0 & 1 & 0 & 0 & 0 & 0 & 0 & 0 & 0 & 1 & 0 & 1 & 0 \\
0 & 0 & 0 & 0 & 0 & 0 & 0 & 1 & 0 & 0 & 0 & 0 & 0 & 0 & 0 & 1 & 0 & 1 \\
0 & 0 & 0 & 0 & 0 & 0 & 0 & 0 & 1 & 1 & 0 & 0 & 0 & 0 & 0 & 0 & 1 & 0 \\
0 & 0 & 0 & 0 & 0 & 0 & 0 & 0 & 0 & 0 & 0 & 0 & 0 & 0 & 0 & 0 & 0 & 0 \\
0 & 0 & 0 & 0 & 0 & 0 & 0 & 0 & 0 & 0 & 1 & 0 & 0 & 0 & 1 & 1 & 0 & 0 \\
0 & 0 & 0 & 0 & 0 & 0 & 0 & 0 & 0 & 0 & 0 & 0 & 1 & 1 & 0 & 0 & 0 & 1 \\
0 & 0 & 0 & 0 & 0 & 0 & 0 & 0 & 0 & 0 & 1 & 1 & 0 & 0 & 1 & 0 & 1 & 0 \\
0 & 0 & 0 & 0 & 0 & 0 & 0 & 0 & 0 & 0 & 0 & 1 & 0 & 1 & 0 & 0 & 1 & 1 \\
0 & 0 & 0 & 0 & 0 & 0 & 0 & 0 & 0 & 0 & 1 & 1 & 1 & 1 & 1 & 1 & 1 & 1 \\
0 & 0 & 0 & 0 & 0 & 0 & 0 & 0 & 0 & 1 & 0 & 0 & 1 & 0 & 0 & 1 & 0 & 0 \\
0 & 0 & 0 & 0 & 0 & 0 & 0 & 0 & 0 & 1 & 1 & 0 & 1 & 0 & 1 & 0 & 0 & 0 \\
0 & 0 & 0 & 0 & 0 & 0 & 0 & 0 & 0 & 1 & 0 & 0 & 0 & 1 & 0 & 1 & 0 & 1 \\
0 & 0 & 0 & 0 & 0 & 0 & 0 & 0 & 0 & 1 & 1 & 1 & 1 & 0 & 1 & 1 & 1 & 0 \\
0 & 0 & 0 & 0 & 0 & 0 & 0 & 0 & 0 & 1 & 0 & 1 & 1 & 1 & 0 & 1 & 1 & 1 \\
0 & 0 & 0 & 0 & 0 & 0 & 0 & 0 & 0 & 1 & 1 & 1 & 0 & 1 & 1 & 0 & 1 & 1 \\
};
\draw[line width=.5pt]
  ([yshift=.2ex] m-21-10.south west) -- ([yshift=.2ex] m-1-10.north west);
\draw[dashed]
  ([yshift=.2ex] m-9-1.south west) -- ([yshift=.2ex] m-9-18.south east);
\end{tikzpicture}\mkern-5mu
\right) \normalsize
\end{equation*}
\caption{Generating matrix and coset representatives for the classical code associated with the $\left(\!\left(9,12,3\right)\!\right)$ CWS quantum code, with the generating matrix for the linear code $D$ above the dashed line and the 12 codewords in $T$ below it.}
\label{cws9}
\end{figure}

Using the $\left(\!\left(9,12,3\right)\!\right)$ CWS code constructed by Yu et al. \cite{Yu2008} as an example, the code has Shor-Laflamme weight enumerators \begin{equation*} A\!\left(z\right)=1+\frac{2}{3}z^{4}+\frac{32}{3}z^{6}+\frac{64}{3}z^{7}+9z^{8} \end{equation*} and \begin{align*} B\!\left(z\right)={ }& 1+68z^{3}+242z^{4}+684z^{5}+1464z^{6} \\ & +1852z^{7}+1365z^{8}+468z^{9}. \end{align*} The nonadditive classical code $C$ associated with the set $\mathcal{TS}$ is constructed using the linear code $D$ generated by $G$ and taking the union of cosets $$C=\bigcup\limits_{i=1}^{12}\left(t_{i}+D\right),$$ where $G$ and $t_{i}$ are defined in Figure \ref{cws9}. Calculating the symplectic distance between each pair of codewords, we find that the distance enumerator \begin{align*} A'\!\left(z\right) ={ }& 1+68z^{3}+242z^{4}+684z^{5}+1464z^{6} \\ & +1852z^{7}+1365z^{8}+468z^{9} \end{align*} is identical to the Shor-Laflamme weight enumerator $B\!\left(z\right)$ of the quantum code.

One interesting observation about the $\left(\!\left(9,12,3\right)\!\right)$ code is that the distance enumerator has all integral-valued coefficients which count the number of codewords of each weight in the code, so $B\!\left(z\right)$ also counts the number of elements of each weight in $\mathcal{TS}$ like a traditional weight enumerator for stabilizer codes. This also holds true for the $\left(\!\left(10,24,3\right)\!\right)$ CWS code constructed by Yu et al. \cite{Yu2007}, but not for the $\left(\!\left(7,22,2\right)\!\right)$ code constructed by Smolin et al. \cite{Smolin2007}, the weight enumerator of which has nonintegral coefficients.

\section{Conclusion}

In this paper we give a combinatorial interpretation for the Shor-Laflamme weight enumerator $B\!\left(z\right)$ for codeword stabilized codes, by connecting the centralizer analogue to a classical code. One question that remains is whether there is a similar combinatorial interpretation for the Shor-Laflamme weight enumerators for nonadditive codes not equivalent to CWS codes. Another question is which CWS codes are similar to the $\left(\!\left(9,12,3\right)\!\right)$ and $\left(\!\left(10,24,3\right)\!\right)$ codes whose weight enumerators $B\!\left(z\right)$ count the number of elements in $\mathcal{TS}$, similar to the case with stabilizer codes.

\bibliographystyle{IEEEtran}

\end{document}